\documentclass[a4paper,11pt,twoside]{amsart} 
\usepackage[T1]{fontenc}
\usepackage[utf8]{inputenc}
\usepackage[english]{babel} 
\usepackage { amsfonts }
\usepackage { amsmath }
\usepackage { amsthm }
\usepackage{amscd}
\usepackage{mathtools}
\usepackage{amssymb}
\usepackage{hyperref}
\usepackage{xcolor}
\usepackage{enumitem}
\usepackage{caption}
\usepackage{tikz-cd}
\usepackage[top=4cm,bottom=4cm,left=3.5cm,right=3.5cm]{geometry}
\usepackage{longtable}
\usepackage{tabularx}
\usepackage{cleveref}
\usepackage{comment} 
\usepackage[ruled,linesnumbered]{algorithm2e}

\usepackage[numbers]{natbib}
\usepackage{subcaption}

\usepackage{tabularx}

\newcommand{\FF}{\mathbb{F}}

\newcommand{\C}{\mathcal{C}}

\newcommand{\zps}{\mathbb{Z}_{p^s}}

\newcommand{\fa}{\mathcal{F}_A}

\newcommand{\rx}{R[x]}

\newcommand{\mo}[1]{\lvert  #1 \rvert}

\newcommand{\good}{$(r,l)$-good }
\newcommand{\nr}{N(R)}

\newcommand{\Z}{\mathbb{Z}}

\theoremstyle{definition}
\newtheorem{definizione}{Definition}[section]
\theoremstyle{definition}
\newtheorem{teorema}[definizione]{Theorem}
\theoremstyle{definition}
\newtheorem{proposizione}[definizione]{Proposition}
\theoremstyle{definition}
\newtheorem{lemma}[definizione]{Lemma}
\theoremstyle{definition}
\newtheorem{ex}[definizione]{Example}
\theoremstyle{definition}
\newtheorem{rem}[definizione]{Remark}
\theoremstyle{definition}

\theoremstyle{definition}
\newtheorem{cor}[definizione]{Corollary}

\theoremstyle{definition}
\numberwithin{equation}{section}

\title{A class of locally recoverable codes over  finite chain rings}

\begin{document}

\author[G. Cavicchioni]{Giulia Cavicchioni}
\address{Department of Mathematics\\
University of Trento \\
Italy
}
\email{giulia.cavicchioni@unitn.it}

\author[A. Meneghetti]{Alessio Meneghetti}
\address{Department of Mathematics\\
University of Trento \\
Italy
}
\email{alessio.meneghetti@unitn.it}

 \author[E.Guerrini]{Eleonora Guerrini}
 \address{LIRMM, Université de Montpellier\\
 France
 }
 \email{eleonora.guerrini@lirmm.fr}


\maketitle
 \begin{abstract}\footnotesize
Locally recoverable codes deal with the task of reconstructing a lost symbol by relying on a portion of the remaining coordinates smaller than an information set. We consider the case of codes over finite chain rings, generalizing known results and bounds for codes over fields. In particular, we propose a new family of locally recoverable codes by extending a construction proposed in 2014 by Tamo and Barg, and we discuss its optimality. The principal issue in generalizing fields to  rings is how to handling the polynomial evaluation interpolation constructions. This leads to deal with  substructif and well conditioned sets in order to find optimal constructions.
\end{abstract}
\section{Introduction}
Introduced in \cite{gopalan}, locally recoverable codes have garnered attention due to their relevance in distributed and cloud storage systems. Data centers and other modern distributed storage systems use redundant data storage to protect against node failures. Indeed they enable local repair of a coordinate by accessing a maximum of $r$ other coordinates. This set of $r$ coordinates is  commonly referred to as the \emph{recovering set} and the code has locality $r$.
Many research efforts have been focused on establishing bounds for the minimum distance and developing construction techniques for locally recoverable codes \cite{gopalan,calde,jin2,jin,tamo,xing}. \\
If  $C$ is a linear code of length $n$, dimension $k$ and locality $r$ over the  field $\FF_q$, then its minimum  distance satisfies \cite{gopalan} 
\begin{equation}\label{lrcf}d \le n- k - \bigg\lceil \frac{k}{r}\bigg\rceil + 2\ .
\end{equation}

A central problem in Coding Theory is to construct optimal codes and to discuss the concept of optimality itself. In this work we say that a code with minimum distance $d$ is  optimal  if no code over the same alphabet and  with the same parameters  has a minimum distance strictly larger than $ d$. A code meeting   bound \eqref{lrcf} is thus an optimal locally recoverable code. 
Constructions of locally recoverable  codes meeting the bound are given in \cite{barg,calde,jin2,jin,tamo,xing}. 
\\ In all these constructions, the $i$-th coordinate together with its recovering set form a 1-erasure correcting code. A possible extension  is presented in \cite{park}, where the authors introduced the $(r, \rho)$-locality,  allowing  recovering $\rho-1$ erasures by looking at other $r$ coordinates.
An additional and relevant generalization can be found in \cite{raw}, where  each coordinate has several pairwise disjoint recovering sets.\\
On the other hand, it's worth noting that the bound proposed in \cite{gopalan} is independent of the alphabet size $q$. An extension in this direction is presented in \cite{cadambe}, where  a bound for the minimum distance of a  locally recoverable codes depending on $q$ is presented.\\

In this paper, we present a generalization of the theory, allowing the alphabet of the  code to be a ring,  rather than a field as in classical Coding Theory.

This paper is organized as follows. In \Cref{sec:1} we recall some basics on linear codes over rings and we introduce locally recoverable codes. In \Cref{sec:LRCbd}, 
 similarly to  \cite{gopalan}, we derive a bound for the minimum distance of a locally recoverable code over a finite chain ring. As in the classical case, the bound is a function of the length, rank, and locality of the code. Additionally, we prove that this bound is not tight for certain values of the parameters of the code. In a similar fashion to \cite{tamo}, in \Cref{sec:3} we construct a family of optimal locally recoverable codes.
The core of this construction lies in the so-called \emph{good polynomials}. In \Cref{sec:galrg}, we build a class of good polynomials over Galois rings.
In \Cref{sec:5}, we insert the construction presented in \Cref{sec:3} into a more general framework. Finally, in \Cref{sec:6}, we explore various  generalizations of the main construction aimed at relaxing some constraints on the code parameters. We finally discuss the maximum possible length of a locally recoverable code over a finite chain ring.

\section{Codes over rings and Locality}\label{sec:1}
\subsection{Generalities on codes over rings}\label{subs:codesring}
Let $R$ be a finite commutative ring. From the structure theorem for finite commutative rings \cite[Theorem VI.2]{mcdonald}, it is well known that $R$ decomposes uniquely (up to the order of summands)  as a finite direct product of $w$ local rings, 
     $R= R_1\times\dots\times R_w $.
In particular, if $R$ is a principal ideal ring, PIR for short, the $R_i $ are finite chain rings. From the decomposition of $R$ we get, for some $n$, that $R^n=R_1^n\times \dots \times R_w^n  . $ \\

From now on, let $R$ be a PIR. A code $C$  of length $n$ over $R$ is a subset $C\subseteq R^n$ and its elements are called \emph{codewords}. However, our focus lies in structured codes.
\begin{definizione}
     An\emph{ $R$-linear code}  of length $n $ is  an $R$-submodule $C \subseteq R^n$. An $R$-linear code $C$ is said to be \emph{free} if $C$ is a free submodule of $R^n$.
\end{definizione}

In the classical framework, $R$ is usually considered to be a Finite Field, and in this case an important parameter of linear codes is the dimension of $C$ as a vector subspace of $R^n$. In our context, we can define instead  the  \emph{type} of the code, denoted by $k$: $$ k=\log_q\mo{C}\ , \text{ with } q=\mo R\ .$$
 \begin{definizione}
Given an $R$-linear code $C$, the \emph{rank} of $C$ is the minimum  $K$ such that there exists a monomorphism $\phi\colon C\to R^K$ as $R$-modules. In addiction, if $\phi$ is an isomorphism, then $C$ is free and $k=K$.
 \end{definizione}

\begin{rem}\label{rem:decoCRT}
    Given  $R = R_1 \times \dots \times R_w$, we define $e_i$ as the element in $R$ represented by $(0, \dots, 1, 0, \dots, 0)$ in $R_1 \times \dots \times R_w$, with $1$ in the $i$-th position. Let $\pi_i\colon R_1^n\times\dots \times R^n_w\to R^n_i \ $ be the $i$-th canonical projection. 
 If $C$ is an $R$-linear code and $c = (c_1, \dots, c_w)\in C$, where $c_i=\pi_i(c)\in R_i^n$,  then the element $$e_ic=(0,\dots,0,c_i, 0,\dots, 0)\in C  \ .$$ Hence, up to isomorphism, $C$ can be uniquely written  as \begin{equation}
     \label{eq:decoco}
 C_1\times\dots\times C_w \subseteq R^n \ \text{ with } 
 \ C_i=\pi_i(C) \  \text{ for all } \  1\le i \le w\ .\end{equation}Therefore, whenever convenient, we  may restrict our focus on codes over local rings.  \end{rem}

A \emph{minimal generating set} of a code $C\subseteq R^n $
is a subset of $C$ that generates $C$ as an $R$-module and it is minimal with respect to inclusion.
If  $R$ is  a finite chain ring, as a consequence of Nakayama's Lemma \cite[Theorem V.5]{mcdonald}, all the minimal generating sets have the same cardinality.  The  cardinality of a minimal generating set of $C$  coincides with the rank of $C$ and therefore we will denote it with $K$.  
 \\  A matrix whose rows form a generating set for the code is  a \emph{generator matrix} for the code. \\

The Hamming metric is a discrete metric counting the number of entries in which two tuples differ, namely, for any $v=(v_1, \ldots ,v_n)$ and $u=(u_1,\ldots,u_n)$ in $R^n$,
$$
\mathrm{d}(v,u)=\left|\left\{i\mid v_i\neq u_i, \;1\leq i\leq n\right\}\right|\;.
$$

The so-called \emph{minimum distance} $d$, i.e.  $$
d=\min_{c_1,c_2\in C\ c_1\ne c_2}\mathrm{d}(c_1,c_2)\;,$$ is a relevant parameter for the code.

Indeed, it is related to the error correction capability, namely, how many coordinates  of a codeword $c$ can be corrupted without compromising our ability of reconstructing $c$ without errors.  If the code is linear, the minimum distance  coincides with the minimum weight of the codewords.\\

It is well known (see for example \cite{macw}) that the Singleton bound holds for any alphabet $R$ of size $q$.
\begin{teorema}
\textbf{(Singleton  bound)} Let $C$ be a  code of length $n$ over an alphabet of size $q$. Then 
\begin{equation*}
     d\le n-\log_{q} \mo C+1\ . 
\end{equation*}
\end{teorema}
If $R$ is a finite chain ring and $C$ is an $R$-linear code of length $n$ and  type $k$, the previous bound reads $$d\le n-k+1 \ .$$
Only free codes can meet this bound and they are said \emph{maximum distance separable} (MDS) codes. However, in the framework of codes over finite chain rings,  the Singleton bound can be improved (see for example \cite{norton}).
 \begin{teorema}\label{generalizedSing} \textbf{(Generalized Singleton bound)} Let $R$ be a finite chain ring and let $C$ be an $R$-linear code of length $n$ and rank $K$. Then 
      \begin{equation*}
     d\le n-K+1 \ .
 \end{equation*}
 \end{teorema}
This bound is generally tighter than the Singleton bound, and they coincide if and only if the code is free. A linear code meeting this bound is said to be \emph{maximum distance with respect to rank }(MDR). \\

For any linear  code  $C\subseteq R^n$ and for any subset $S\subset \{1,\dots,n\}$ of the coordinates, we  define $C_S$ to be the \emph{punctured code of $C$ in $S$}, obtained by deleting in each codeword  all but  the coordinates indexed in $S $. If $\mo C=\mo{C_S}$ then $S$ is an  \emph{information set of size $\mo S$ } for $C$. In the following, we will denote with $\kappa$ the minimal size of an information set. 
Note that for codes over finite chain rings  the minimum size of an information set coincides with the rank and $\kappa=K$. 
\begin{cor}\label{cor:singbd}
    Let $C$ be a code with minimum distance $d$ and let $S$ be a subset of coordinates which does not form an information set. Then $$\mo{S}\le n-d\ .$$
\end{cor}
\begin{proof}
    By contradiction, assume $\mo S\ge n-d+1$.  Since $S$ is not an information set $\mo{C_S}<\mo{C}$.  Note that $C_S$ is obtained  form $C$ by removing at most $d-1$ coordinates: this contradicts the definition of minimum distance.
\end{proof}
\subsection{Locally recoverable codes}

The goal of a local recovery technique is to enable the retrieval of lost encoded data using only a small portion of the available information, rather than requiring access to the complete codeword $c$. \\
Let $R$ be a finite commutative ring. 

    \begin{definizione}
    Let $C$ be a (possibly non-linear)  code in $ R^n$ and let $(c_1,\dots,c_n)$ be a codeword. We say that the coordinate  $i\in \{1,\dots,n\}$ has \emph{locality} $r$ if  there exists a subset $S_i\subseteq \{1,\dots,n\}\setminus \{i\}$   such that: \begin{itemize}
        \item \emph{(locality)} \hspace{0.2cm}$\mo{S_i}\le r$, 
        \item  \emph{(recovery)} $\mo{C_S}=\mo{C_{S\cup \{i\}}}$.
    \end{itemize} 
     $C$ is a \emph{locally recoverable code} (LRC) with \emph{locality $r$} if each coordinate has locality $r$.
    \end{definizione}

In other words, any symbol $c_i$ of any codeword $c$ can be recovered by accessing at most $r$ other symbols of $c$. 
If we are presented with a codeword $c$ that is error-free except for an erasure at position $i$, we can retrieve the original codeword by only examining the coordinates in $S_i$. For this reason, $S_i$ is referred to as a \emph{recovering set} for  $i$. Moreover we will say that  $S_i\cup \{i\}$ is a \emph{dependent set}.\\  
If  $R$  is a finite chain ring and $C$ is an $R$-linear code of length $n$, rank $K$ and locality $r$, we will say that $C$ is an $(n,K,r)$-code. 

Of course, one can choose $R=\FF_q$. In this case we recover the classical theory of locally recoverable codes over Finite Fields.


In this work we say that an $(n, K, r)$-code  over $R$ with minimum distance $d$ is an \emph{optimal locally recoverable code} if no $(n, K, r)$-code  over $R$ has a minimum distance strictly larger than $d$.
\\ 

In 2014 Tamo and Barg \cite{tamo}  presented a clever construction for optimal locally recoverable codes based on  polynomial interpolation. 
In the following sections we will extend this construction in the more general framework of codes over finite chain rings.

\section{Lower bound on the minimum distance of a locally recoverable code}\label{sec:LRCbd}
Let $R$ be a commutative ring, let $C$ be a code of length $n$ over  $R$ and let $\kappa$ be the minimum size of its information sets. From now on, we will denote by $S_i$ a recovering set for the coordinate $i$. 


For any code $C$, the set of  dependencies involving at most $r+1$ coordinates defines a directed graph. If $S_i$ is a recovering set for $i$, we define $X$  to be the graph whose vertex set is the set of coordinates $\{1,\dots,n\}$ and in which there exists a directed edge from $i $ to $j$ if and only if $j\in S_i$. For a vertex $v$ we will denote by $N(v)$ the  outgoing neighbors of $v$.
\begin{teorema}\label{thm:lrcbdgen}
    Let $C $ be a code of length  $n$ and locality $r$ over $R$. Then the minimum distance satisfies 
    \begin{equation}\label{bd:loc}
        d\le n-\kappa-\bigg\lceil\frac{\kappa}{r}\bigg\rceil+2 \ . 
    \end{equation}
    Moreover
    \begin{equation}\label{eq:rate}
        \frac{\kappa}{n}\le \frac{r}{r+1}\ .
    \end{equation}  
\end{teorema}

\begin{proof}
    Let $X$ be the directed graph associated to the code $C$. Note that the outgoing degree of each vertex is at most $r$. A modification of Turàn Theorem on the size of the maximal independent set in a graph, \cite[Theorem A.1]{tamo}, establishes that $X$ contains an induced acyclic subgraph  $X^\mathcal{U}$ on the vertex set $\mathcal{U}$ with $$\mo{\mathcal{U}}\ge \frac{n}{r+1}\ .$$Let $i$ be a coordinate without outgoing edges: $i $ is a function of the coordinates $\{1,\dots,n\}\setminus\mathcal{U}$. The induced subgraph of $X$ on $\mathcal{U}\setminus \{i\}$ is a directed acyclic subgraph. Let $i'$ be a vertex without outgoing edges in $X^{\mathcal{U}\setminus \{i\} }$: $i' $ is a function of the coordinates $\{1,\dots,n\}\setminus\mathcal{U}$. Iterating, we conclude that any coordinate $i\in \mathcal{U}$ is  a function of  the coordinates $\{1,\dots,n\}\setminus\mathcal{U}$. Therefore, there are at least $\mo{\mathcal{U}}\ge \frac{n}{r+1}$ redundant coordinates. Thus the number of the information coordinates $\kappa$ is at most 
    $\kappa\le n-\frac{n}{r+1}=\frac{nr}{r+1}\ .$
\\   

To establish the bound on the minimum distance, we first build a large set $T\subseteq\{1,\dots,n\}$ which does not form an information set. Then, we use \Cref{cor:singbd} to complete the proof. Let $M(T)$ be the number of independent elements in $T$. \\
Algorithm \ref{algo:maxs} constructs the desired set $T$.\\

    \begin{algorithm}\caption{ Construction of $T$}\label{algo:maxs}
     \LinesNumbered
    
    Let $i=0$, $T_0=\{\}$.
    
    \BlankLine
    \While{$M(T_{i-1})\le \kappa-2$}{
       Pick $j\in \{1,\dots,n\}\setminus T_{i-1}$ such that $j$ has at least one outgoing edge in $V_{ \{1,\dots,n\}\setminus T_{i-1}}$.
  
    \If{$M(T_{i-1}\cup N(j))< \kappa$}{Set $T_i=T_{i-1}\cup N(j)\cup j$;}
    \Else{pick $N'(j)\subset( N(j)\cup j)$ so that $M(T_{i-1}\cup N'(j))=\kappa-1$ and set $T_i=T_{i-1}\cup N'(j)$.}
    $i=i+1$.
  }
\end{algorithm}

 Since the  cardinality of an information set is at least  $\kappa$, there exists $j $ as desired in Line 3 of Algorithm \ref{algo:maxs},. Let $h$ denote the number of steps of the algorithm, let \begin{equation*}
    t_i=\mo{T_i}-\mo{T_{i-1}}, \quad \mo{T_h}=\sum_{i=1}^h t_i \ ,
\end{equation*} and \begin{equation*}
    m_i=M(T_i)-M(T_{i-1}), \quad M(T_h)=\sum_{i=1}^h m_i=\kappa-1\ .
\end{equation*}

There are two possible cases to consider: one where the \textbf{else} condition in Line 7 is reached, and the other where it is never executed.
\begin{enumerate}
    \item[Case 1.] Assume $M(T_{i-1}\cup N(j))\le \kappa-1$ for all $1\le i\le h$. In each step we add $t_i\le r+1 $ coordinates. Moreover, $m_i\le t_i -1\le r$. Since in the last step we have $\kappa-1$ independent coordinates, the number of steps is at least $ \lceil\frac{\kappa-1}{r}\rceil $. Thus \begin{equation*}
        \mo{T}=\sum_{i=1}^h t_i\ge \sum_{i=1}^h(m_i+1)\ge \kappa-1 +h\ge\kappa-1 +\bigg\lceil\frac{\kappa-1}{r}\bigg\rceil \ . 
    \end{equation*}
    Since $\kappa-1+\lceil\frac{\kappa-1}{r}\rceil \ge \kappa+ \lceil\frac{\kappa}{r}\rceil-2$, we get the claim. 
    \\
    
    \item[Case 2.] Since in the last step we have hit the condition $M(T_{h-1}\cup N(j))=\kappa$ and $M$ increases at most by $r$ per step, then $h\ge \lceil\frac{\kappa}{r}\rceil$. For any $1\le i \le h-1$ we add at most $r+1 $ coordinates and $m_i\le t_i-1.$ Since $M(S_{h-1})\le \kappa-2 $,  $m_h\ge 1 $ and $t_h\ge m_h$. Therefore \begin{equation*}
        \mo{T}=\sum_{i=1}^h T_i\ge \sum_{i=1}^{h-1}(m_i+1)+m_h\ge \kappa-1+h-1\ge\kappa+\bigg\lceil\frac{\kappa}{r}\bigg\rceil -2\ . 
    \end{equation*}
\end{enumerate}
\end{proof}
Let $R$ be a finite chain ring,  let $\gamma$ be the generator of the maximal ideal with  nilpotency index $s$.  As shown in \cite[Proposition 3.2]{norton2}, any $R$-linear code $C$ is permutation equivalent to a code having the following  generator matrix in standard form:
	\begin{equation*}\label{genma}
G=\begin{bmatrix}
I_{k_0}&A_{0,1} &A_{0,2} &A_{0,3}&\dots& A_{0,s-1}& A_{0,s} \\
0 &\gamma I_{k_1} & \gamma A_{1,2}& \gamma A_{1,3}& \dots& \gamma A_{1,s-1}&\gamma A_{1,s}\\
0 &0 & \gamma^2 I_{k_2} & \gamma^2 A_{2,3}& \dots& \gamma^2 A_{2,s-1}&\gamma^2 A_{2,s}\\

\vdots &  \vdots& \vdots & \vdots && \vdots &\vdots \\
0 &0 & 0 & 0 & \dots & \gamma^{s-1}I_{k_{s-1}}& \gamma^{s-1} A_{s-1,s}\\
\end{bmatrix}\ ,
\end{equation*} where $A_{i,s}\in M_{k_i\times n-K}(R/\gamma^{s-i}R)$ and $A_{i,j}\in M_{k_i\times k_j}(R/\gamma^{s-i}R)$ for $j< s$.
 The  parameters  $k_0,\dots, k_{s-1}$ are the same for all generator matrices in standard form, and  $ C $ is said to be of \emph{subtype} $(k_0,k_1,\dots,k_{s-1}) $.  
\\ 

In the  framework of codes over finite chain rings bound \eqref{bd:loc} reads:
\begin{cor} \textbf{(LRC bound for $R$-linear codes)}
    Let $R$ be a finite chain ring and let $C$ be an $R$-linear code of length $n$, rank $K$ and locality $r$. Then 
     \begin{equation}\label{bd:locfcr}
        d\le n-K-\bigg\lceil\frac{K}{r}\bigg\rceil+2 \ .
    \end{equation}
 \end{cor}

For linear codes over rings, the  dependence relations among the columns of $G$ can serve as recovering sets. 
However, opposed to the case of vector spaces, the notion of linear   independence for modules over rings is not well defined. Indeed, for a finite chain  ring $R$, the following two definitions are not equivalent.
\begin{definizione}
    The vectors $v_1,\dots,v_u\in R^n$ are said to be \emph{modularly independent} over $R$ if  $\sum_{i=0}^u s_i v_i=0 $ with $s_i\in R$ implies $s_i$ is not a unit for all $i$.
\end{definizione}
In particular the vectors $v_1,\dots,v_u\in R^n$ are  modularly independent if none of them can be written as a linear combination of the others. 
\begin{definizione}
    The non-zero vectors $v_1,\dots,v_u\in R^n$ are said to be \emph{linear independent} over $R$ if $\sum_{i=0}^u s_i v_i=0 $, $s_i\in R$ implies $  s_i=0$  for all $i$.
\end{definizione}
Therefore the vectors $v_1,\dots,v_u\in R^n$ are  linear independent if  the only linear combination   of the $v_i$ to 0 is given by setting all the scalars to zero.  \\ For further details on this topic refer to  \cite{park}.\\

Hence, the modular dependencies relations allows to gain local recoverability. \\

Note that each symbol  in an $R$-linear code of rank $K$ has locality at most $K$. Thus $r$ satisfies $1\le r\le K$. In particular:
     \begin{itemize}
         \item If $r=K$, the LRC bound  reduces to the generalized Singleton bound  and optimal LRC codes are MDR codes;
         \item If $r=1$,  bound \eqref{bd:locfcr} reads $$d\le n-2K+2=2\bigg(\frac{n}{2}-K+1\bigg) \ . $$         
     \end{itemize}
Therefore, by replicating each symbol twice in an MDR code  of length $\frac{n}{2}$ and rank $K$, we get an optimal linear code with locality $r=1$. 

\begin{rem}
If $C$ is an  $R$-linear code of subtype  $(k_0,k_1,\dots,k_{s-1}) $, following the same steps of \cite[Theorem 3.2]{forbes}, we obtain an upper bound on the minimum distance: \begin{equation} \label{bd:type}
 d \le n- k - \bigg\lceil \frac{k}{r}\bigg\rceil + 2\ , \ \text{ where }\  k=\frac{1}{s}\sum_{i=0}^{s-1} (s-i)k_i \ . \end{equation}
 Note that bound \eqref{bd:locfcr} is  in general tighter  then \eqref{bd:type} and the two inequalities coincide if and only if the code is free. 
\end{rem}

Codes that attain the LRC bound on finite chain rings can be used as building blocks to construct codes that achieve the LRC bound on finite PIRs.
\begin{lemma}
Let $R=R_1\times\dots\times R_w$ be a PIR and let $C=C_1\times\dots\times C_w \subseteq R^n$ be an $R$-linear code. If $K$ and  $K_i$ are the ranks of $C$ and $C_i$ respectively, then: 
\begin{enumerate}
    \item $\mathrm{d}(C)=\min_i \ \mathrm d(C_i)$;
    \item  $K=\max_i \  K_i$.
\end{enumerate}
\end{lemma}
\begin{proof}\
 To prove the first statement, let $c=(c_1,\dots,c_w)\in C$ be a codeword. In accordance with  \Cref{rem:decoCRT} $e_ic=(0,\dots, 0,c_i,0,\dots,0)\in C$ and hence the claim.\\  For the second claim, let $\varphi_i\colon{C_i}\to R^{K_i}$ be the monomorphism defining the rank of the code. By composing  $\varphi_i$ with the canonical embedding,  
$$\psi_i\colon C_i\to R^{K_i}\hookrightarrow R^K \ ,$$ we get another monomorphism. Let $$(\psi_1,\dots,\psi_w)\colon C_1\times\dots\times C_w \to R_1^K\times\dots\times R^K_w \ .$$ Since $C= C_1\times \dots\times C_w $ and $R^K= R_1^K\times\dots\times R^K_w $, $(\psi_1,\dots,\psi_w) $ induces a monomorphism $\psi\colon C\to R^K $. According to the definition of $\psi_i$, $K$ is the minimum integer ensuring the injectivity of the map $\psi$  and the claim follows. 

\end{proof}

\begin{teorema}
Let  $R=R_1\times\dots\times R_w$  be a finite PIR and let   $C=C_1\times\dots\times C_w \subseteq R^n$  be an $R$-linear code.  If $C_i$ is an optimal LRC over $R_i$ for all $1\le i\le w$, then  $C$ is optimal LRC over $R$.
\end{teorema}

\begin{proof}
    \begin{align*}
        \mathrm d(C)&=\min_{1\le i\le w}  \ \mathrm d(C_i)=\min_{1\le i\le w}  \ n-K_i -\bigg\lceil \frac{K_i}{r}\bigg\rceil+2=\\&= n-\max_{1\le i\le w} \bigg\{K_i+\bigg\lceil \frac{K_i}{r}\bigg\rceil\bigg\} +2= n-K-\bigg\lceil\frac{K}{r}\bigg\rceil+2 \ . 
    \end{align*}
\end{proof}

Hence, we can focus our studies on LRC codes over finite chain rings.

\subsection{Non-existence of  $R$-linear codes achieving the LRC bound for certain parameters}\label{sec:nonex}
The aim of this section is to show that codes achieving the LRC bound do not exist for all possible values of $n, \  K$ and $r$. To do this,  we  will introduce a weaker notion of locality: the  information locality.\\

Let $R$ be a finite chain ring.
\begin{definizione}
    The code $C\subseteq R^n$ has \emph{information locality} $\bar r$ if there exists an information set $I\subset\{1,\dots,n\}$ such that any information coordinate $i\in I$ has locality as most $\bar r$.
\end{definizione}

Following the same steps of \Cref{thm:lrcbdgen} one can prove that  the minimum distance of  an  $R$-linear code $C$ of length $n$, rank $K$ and information locality $\bar r$ is bounded by  $$ d \le n-K-\frac{K}{\bar r}+2 \ .  $$  In the following, we will denote by  $\overline 
X_C$ the directed graph defined by the modular dependencies involving at most $\bar r+1$ coordinates of   $C$.
\begin{teorema}\label{thm:infolocgraph}
    Let $C$ be an $R$-linear code of length $n$, rank $K$ and with information locality $\bar r$. Suppose $K\mid \bar r$ and \begin{equation}\label{eq:distinfo}
        d=n-K-\frac{K}{\bar r}+2 \ .
    \end{equation}   $\overline X_C$  has at least $\frac{K}{\bar r}$  connected components with exactly $\bar r+1$ vertices.  
\end{teorema}
\begin{proof}

Algorithm \ref{algo:maxs} yields two sequences $\{t_i\}_{1,\dots, h}$ and $\{m_i\}_{1,\dots, h}$.
    \begin{itemize}
        \item[Case 1.] If $\bar r=1$, since $m_i\le 1$, the \textbf{else} block (Line 7, Algorithm \ref{algo:maxs}) is never executed. Therefore,  \begin{align*}
            \mo{T}=\sum_{i=0}^h t_i\ge \sum_{i=0}^hm_i+h\ge K-1+K-1=2(K-1)=2(n-d) \ ,
        \end{align*}
        where the last equality follows form \eqref{eq:distinfo}.  From \Cref{cor:singbd} we get $\mo T=n-d=2(K-1)$. Since $\sum_{i=0}^h m_i=K-1$ then $ h=K-1, \ t_i=2,$ and $m_i=1$ for all $1\le i\le h$. Therefore there are at least $K-1$ connected components of size 2. 
         \item[Case 2.]If $\bar r\ge 2 $ and $\bar  r\mid K $ then $K\not\equiv 1 \mod \bar r$ and $K-1+\lceil\frac{K-1}{\bar r}\rceil\ge K+\frac{K}{\bar r}-2
         $. Therefore, in order to find a lower bound on $\mo T$, we may assume the \textbf{else} condition (Line 7, Algorithm \ref{algo:maxs}) is executed. \begin{align*}
             \mo T=\sum_{i=0}^h t_i\ge \sum_{i=0}^{h} m_i+h\ge K-1 \frac{K}{\bar r} -1 =n-d\ .
         \end{align*}
         From \Cref{cor:singbd} we get $\mo T=n-d=K+\frac{K}{\bar r}-2$, and hence we always enter in  the \textbf{else} block.  Since $ \sum_{i=0}^{h} m_i=K-1$, then $h=\frac{K}{\bar r}$. Moreover $m_i\le\bar r$ implies $m_j= \bar r-1$ for some $j\in \{1,\dots,h\}$  and $ m_i= \bar r$ for all $i\ne j$. In particular $j=h$, otherwise the \textbf{else} condition is never executed.
         \\ First suppose there exists a connected component with $\bar r$ vertices. By adding this component to $T$ in the first step, we would get $m_1\le \bar r-1$. Finally, assume there are $\frac{K}{\bar r}-1$ connected components, namely, the recovering set of $j, l\in \{1,\dots,n\}$ intersects. Let $S_j $ and $S_l$ be the recovering sets of $j$ and $l$ respectively and let $S=S_j\cup S_l$. Note that the  number of modularly independent coordinates in $S$ is at most $2\bar r-1$. By including $S_j$ and $S_l$ in the set $T$ at the beginning of the algorithm, we ensure that  $m_1+m_2\le 2\bar r-1$. Both cases lead us to a contraction that prevents the \textbf{else} condition from being executed.

    \end{itemize}
\end{proof}

\begin{teorema}

    Let $C$ be an $R$-linear code of length $n$, rank $K$. If there exists an information set whose information locality is $\bar{r}\mid K$ and  $C$ has minimum distance $ d=n-K-\frac{K}{\bar{r}}+2$ with $d\leq \bar{r}+2$, then some redundant coordinates have locality $r>\bar{r}$.
    
\end{teorema}

\begin{proof}
Let $I$ be an information set with information locality $\bar r$.
    We prove that $\overline X_C$, the directed graph associated to $C$, has exactly $\frac{K}{\bar r}$ connected components. \\ From \Cref{thm:infolocgraph} we know that the number of connected components is at least $\frac{K}{\bar r}$, we now show they are exactly $\frac{K}{\bar r}$.\\ Let $m$ be the number of connected components. By contradiction assume $m\ge \frac{K}{\bar r}+1$. \begin{equation*}
        n\ge m(\bar r+1)=K+\frac{K}{\bar r}+\bar r+1>K+\frac{K}{\bar r}+d-2 \ ,
    \end{equation*} which contradicts the choice of $n$. Therefore $\overline X_C$ must contain $n-\frac{K}{\bar r} (\bar r+1)=d-2$ isolated vertices which do not participate to any modular relations. \\
    The same argument applies to every choice of an information set and its associated information locality.
    
\end{proof}
\begin{cor}
    Let $C$ be an $(n,K,r)$-code. If $r|K$ and $r+\frac{K}{r}>n-K-1$, then $C$ does not achieve the LRC bound. 
\end{cor}

\section{ Extending the Tamo-Barg construction over finite chain rings}\label{sec:3}

The  construction by Tamo and Barg, \cite{tamo}, allows to obtain  optimal LRC codes over finite fields using particular types of polynomial, the so-called \emph{good polynomials}. Polynomial interpolation is used in order to recover erased data.

\subsection{Polynomials over rings}
It is important to have in mind that polynomials over rings lack some desirable properties of  polynomials over fields. For example, when considering a ring $R$, the evaluation map  $$ev_\alpha\colon \rx\to R\ , \quad \quad f\mapsto f(\alpha)$$ is an homomorphism if and only if $R$ is commutative. Moreover, in this framework,  polynomial interpolation problems also require a  greater attention.\\

If  $R$ is a local ring with maximal ideal $M$ and residue field $K=R/M$,  we will denote by  $\Bar{y}$  the image of $y   \in R $ under the canonical projection from $R$ to $K$. In addiction, for a set $T\subseteq R$ we define $\overline T=\{\bar t \mid t\in T\}$.  \\

Let $N(R)$ denote the group of units of $R$.
\begin{definizione}
    A subset $T\subseteq \nr$ is said to be \emph{subtractive} if, for all distinct $a,b\in T, \ a-b\in \nr$.
\end{definizione}

Here is a simple property of local rings. 
\begin{lemma}\label{lem:projss}
Given $r,s \in R$, then $\Bar{r}\ne\Bar{s}$ if and only if $r-s\in \nr$.\end{lemma} 
Therefore $T$ is a subtractive subset of $R$ if and only if $\mo{T}=\mo{\overline{T}}$.     \\

In line with \cite{armand}, we  provide the definition for well-conditioned sets.
\begin{definizione}
    A set $\{a_1,\dots, a_n\} $ is \emph{well-conditioned} in $R$ if one of the following conditions is satisfied:
    \begin{enumerate}
        \item $\{a_1,\dots, a_n\}$ is subtractive in $\nr$;
       \item For some $i$, $\{a_1,\dots a_{i-1},a_{i+1},\dots, a_n\} $ is subtractive in $\nr$ and either $R$ is local and $a_i$ is a zero-divisor or $a_i=0$. 
    \end{enumerate}
\end{definizione}
Polynomial reconstruction remains valid if we restrict to well-conditioned sets. \begin{proposizione}\label{prop:nradici} \cite[Corollary 9]{quintin}
   Let $f\in \rx$ be a polynomial of degree at most $n-1$ with at least $n$ roots in a well-conditioned set of $R$. Then $f=0$. 
\end{proposizione}
\begin{cor}\cite[Corollary 10]{quintin}
     Let $\{a_1,\dots,a_n\}$ be a well-conditioned set in  $R$ and let $\{y_1,\dots,y_n\}$ be a subset of $R$. Then there exists a unique polynomial $f\in\rx$ of degree at most $n-1$ such that $f(a_i)=y_i$ for all $1\le i \le n$. 
\end{cor}
\Cref{prop:nradici} points out that, unlike polynomials over fields, the number of roots of a polynomial over a ring is not bounded by its degree.
Nonetheless, for polynomials over local rings, there exists a bound on the number of roots, which depends on the polynomial's degree. The following Corollary is a consequence of the Hensel lifting \cite[Chapter XIII, Section (C)]{mcdonald}.
\begin{cor}\label{cor:nroots}
Let $f(x)\in \rx$ be a polynomial of degree $n$. The number of roots of $f$ in $R$ is at most $np^{(s-1)m}$.
    
\end{cor}

\subsection{Code construction}
Let $R$ be a finite chain  ring with $   \mo R=q$.   Given $f \in \rx$, if $f$ is constant on the set $A$,  we will denote by $f(A)$ the value of $f$ on $A$.  

From now on, we will refer to a polynomial whose leading coefficient is a unit as a monic polynomial.
 \begin{definizione}\label{def:goodp}Let $l\in \mathbb{N}^+$ and $A_1,\dots, A_l$ pairwise disjoint subsets of $R$ of size $r+1$. A polynomial $g \in R[x]$   such that:  \begin{itemize}  
    \item Its degree is $r+1$; 
    \item It is monic; 
    \item  It is constant on $A_i$, i.e., for any $1\le i \le l, \ g(A_i)=c_i$ with $c_i\in R$;   \end{itemize} is said to be  $(r,l)-$\emph{good} on  the blocks $A_1,\dots, A_l$. \end{definizione}     

\begin{teorema}\label{teo:tbcodes}
    Let $r\ge 1$ and let $A_1,\dots,A_l$ be subsets of $R$ such that $A=\bigcup_{i=1}^l A_i$ is well-conditioned. Let $g(x)\in\rx$ be an \good polynomial on the blocks of the partition of $A$. For $t\le l $, set $n=(r+1)l$ and $K=rt$. Let $$a=(a_{i,j}, \ 0\le i\le r-1, \ 0\le j\le t-1)\in R^K \ .$$ We define the \emph{encoding polynomial}
    \begin{equation}\label{eq:encpoly}
        f_a(x)=\sum_{i=0}^{r-1}\sum_{j=0}^{t-1} a_{i,j}g(x)^jx^i \ ;
    \end{equation}
    and the code \begin{equation*}
   \C=\bigg\{ (f_a(\alpha), \ \alpha\in A) \mid a\in R^K\bigg\} \ .
    \end{equation*}
    Then $\C$ is a free $(n,K,r)$-code with minimum distance  $d=n-K- \frac{K}{r}+2$. Hence  $\C$ is an optimal locally recoverable code.
\end{teorema}

\begin{rem} In the following, we  provide an overview of the technique we will use to recover an erased symbol. 
Let $a\in R^K$ be the message vector and assume $(f_a(\gamma), \ \gamma\in A)$ is sent. Suppose that  the symbol corresponding to the location $\alpha\in A_j$ is erased and let $c_\beta$ for all $\beta\in A_j\setminus\{\alpha\}$ represent the remaining $r$ symbols in the locations of set $A_j$. Since $g$ is an $(r,l)$-good polynomial on the block of the partition of $A$,  $f_a(x)$ is a polynomial of degree at most $r-1$ when restricted to $A_j$. Hence, in order to find $c_\alpha=f_a(\alpha)$, we find the unique polynomial $\delta(x)$ of degree less than $r$ such that $\delta(\beta)=c_\beta$ for all  $ \beta \in A_j\setminus\{ \alpha\}$ and we set $c_\alpha=\delta(\alpha)$. The polynomial  $\delta(x)$ is called  the \emph{decoding polynomial} for $c_\alpha$. 
\end{rem}
\begin{proof}
    \begin{itemize}
        \item \textbf{Type of the code:} Recall that $g$ is monic and of degree $r+1$. Therefore, for $i=0,\dots,r-1$ and $j=0,\dots t-1$ the $K$ polynomials $g(x)^jx^i$ are all of distinct degrees. Suppose $f_a(x)=f_b(x)$ for some $a\ne b$. Then 
        \begin{align*}
            f_a(x)-f_b(x)=\sum_{i=0}^{r-1}\sum_{j=0}^{t-1} (a_{i,j}-b_{i,j})g(x)^jx^i=0
        \end{align*}
        
        if and only if $a_{i.j}=b_{i,j}$ for all $i$ and $j$, and in this case the map $a\mapsto f_a$ is injective. On the other hand, by \eqref{eq:encpoly}, the degree of $f_a(x)$ is bounded by $$ (t-1)(r+1)+r-1=K+\frac{K}{r}-2\le n-2\ , $$ where  the last inequality comes from \eqref{eq:rate}. Since the set of evaluation point is well-conditioned, $f_a$ and $f_b$ give rise to two distinct codewords and $\mo{\C}=q^K$. Therefore $\C$ is of type $K$.
        \item\textbf{Minimum distance and rank:} Since the set of evaluation points $A=\bigcup_{i=1}^l A_i$ is well-conditioned, the number of zeros of $f_a(x)$ is bounded by its degree.  The encoding is linear and hence $$d\ge n-\max_a \ \text{deg}(f_a)=n-K- \frac{K}{r}+2. \ $$ On the converse, let $\overline{K}$ be the rank of $\C$. By \eqref{cor:singbd}, $$ d\le n-\overline{K}- \frac{\overline{K}}{r}+2\le n- K- \frac{K}{r}+2 \ .$$ Therefore $d=n- K-\lceil \frac{K}{r}\rceil+2$, $K=\overline{K}$ and the code is free. 
        \item  \textbf{Locality:}  Assume that the symbol $c_\alpha=f_a(\alpha)$ corresponding  to the location $\alpha\in A_j$ is lost. Let $$ f_i(x)=\sum_{j=0}^{t-1}a_{i,j}g(x)^j\ . $$  Since $g$ is constant  on the sets $A_j$, $f_i$ is also constant on $A_j$. If $$\delta(x)=\sum_{i=0}^{r-1}f_i(\alpha)x^i \ , $$ then  $$\delta(\beta)=\sum_{i=0}^{r-1}f_i(\alpha)x^i=\sum_{i=0}^{r-1}f_i(\beta)x^i= f_a(\beta)\ ,$$ namely, $f_a(x) $ and $\delta(x)$ coincides on the locations of $A_j$. Since the degree of $\delta(x)$ is at most $r-1$ on a subtractive subset $A_j$, $\delta$ can be interpolated from the $r$ symbols $c_\beta$ for $\beta\in A_j\setminus \{\alpha\}$. Finally $c_\alpha$ is obtained by computing $\delta(\alpha)$.
    \end{itemize}
\end{proof}
\begin{rem}\label{rem:rdivk}
We have that:
    \begin{enumerate}
        \item If $r=K$ the construction does not require good polynomials and reduces to Reed-Solomon codes.
            \item Analogously to the classical case \cite[Section~3A]{tamo}, the construction can be generalized even for the case $r\nmid K$.
    \end{enumerate}
\end{rem}

Notice that the assumption that a good polynomial must be monic is unnecessary. If we remove it in \Cref{def:goodp}, following the same steps of  \Cref{teo:tbcodes}, we obtain a non-free code  with the same parameters but having a smaller size.  
\begin{ex}
 In the following, we construct an optimal code over $\mathbb{Z}_{11^2}$ with  length  $n=10$,  rank  $K=8$ and  locality  $r=4$. Since $r+1=5$, we need to find a polynomial $g(x)$ of degree 5 which is constant on 2 disjoint sets of size 5. If \begin{equation*}
     A=A_1\cup A_2 \ \text{ with }\ A_1=\{1,3,9,27,81\} \ \text{ and } \ A_2\{40,94,112,118,120\} \ ,
 \end{equation*} then the polynomial $g(x)=x^5$ is constant on $A_1$ and $A_2$: \begin{equation*}
     g(1)=g(3)=g(27)=g(81)=1  \text{ and } g(40)=g(94)=g(112)=g(118)=g(120)=120 \ . 
 \end{equation*} If $$a=(a_{0,0},\ a_{0,1},\ a_{1,0},\ a_{1,1},\ a_{2,0},\ a_{2,1},\ a_{3,0},\ a_{3,1})$$  is the message vector, the encoding polynomial associated to $a$ reads $$f_a(x)=a_{3,1}x^8+a_{2,1}x^7+a_{1,1}x^6+a_{0,1}x^5+a_{3,0}x^3+a_{2,0}x^2+a_{1,0}x+a_{0,0} \ . $$ Since $A$ is subtractive and $\deg f_a(x)\le 8$ we have $d_C\ge 2$ and, by the LRC bound \eqref{bd:locfcr}, $d_C=2$. If $$\bar a = (1, 0,3, 7,0,0,11, 1)\ ,$$ is sent then the  encoding polynomial associated to $\bar a $ is $$ f_{\bar a }=x^8+7x^6+11x^3 +3x+1 \ .$$ The codeword corresponding to $\bar a$ is found to be $$c=(23,113,6,33,72,114,116,106,7,25)\ .$$ Suppose $c$ is sent
 $$y=(23,113,6,33,\times,114,116,106,7,25)\ $$
is received.  The fifth coordinate has been erased and  $f_{\bar a}(81)$ is unknown. \Cref{teo:tbcodes} ensures that it can be recovered just by accessing to the first 4 codeword symbols.  After having computed the decoding polynomial $\delta(x)=12x^3+10x+1$, we can find the missing value  $\delta(81)=72$. 
\end{ex}

\section{Construction of good polynomials over Galois ring}\label{sec:galrg}

Good polynomials play a fundamental role in the previous construction, therefore it becomes crucial to produce good polynomials together with the partition of the set $A$. It is known that classes of good polynomials over finite fields exist \cite{liu,micheli,tamo}. In particular, Micheli in   \cite{micheli} introduced a framework that allows the generation of good polynomials over finite fields.The natural question that arises now is whether there exist good polynomials over rings which are not fields.  \\ Indeed, it is true. Here we construct a class of good polynomials over Galois rings exploiting the structure of its group of units.
\begin{definizione}
    Let $p$ be a prime, $s,m$ positive integers. The \emph{Galois ring} $GR(p^s,m)$ of characteristic $p^s$ and with $p^{sm}$ elements is the quotient ring $$GR(p^s,m)\cong \Z_{p^s}[x]/(f) \ , $$ where $f\in \zps[x]$ is a monic irreducible polynomial of degree $m$ such that $\bar{f}$ is irreducible in $\Z_p$, where $\bar f $ denotes the image of $f$ under the canonical projection.
\end{definizione}
A Galois ring $GR(p^s,m) $ is a local ring with maximal ideal $M=( p) $ and whose residue field $F=GR(p^s,m)/M $ is isomorphic to the finite field $\FF_{p^m}$. Its group of units  has order $(p^m-1)p^{m(s-1)}$.  
\begin{teorema}(\cite[Theorem XVI.9]{mcdonald})
    Let $R=GR(p^s, m)$. Then $$\nr=G_1\times G_2 \quad \text{ where }$$ 
    \begin{itemize}
        \item $G_1$ is a cyclic group of order $p^m-1$;
        \item $G_2$ is a group of order $p^{(s-1)m}$ such that
        \begin{itemize}
            \item if $p$ is odd or $p=2$ and $s\le 2$, $G_2$ is a direct product of $r$ cyclic groups of order $p^{s-1}$;
            \item if $p=2$ and $s\ge 3$, $G_2$ is a direct product of a cyclic group of order $2$, a cyclic group of order $2^{s-2}$ and $m-1$ groups of order $2^{s-1}$.
        \end{itemize}
    \end{itemize}
\end{teorema}
Therefore, there is a unique maximal cyclic subgroup of $\nr$ having order relatively prime to $p$ (namely $p^m-1$).

\begin{lemma}(\cite[Lemma XV.1]{mcdonald})\label{thm:mc1}
   Let $f\in \rx$ be a polynomial which is not a zero divisor. Suppose $\bar f $ has a zero $ s\in F$. Then $f$ has one and only one zero $r$ such that $\bar r=s$.
\end{lemma}
\begin{proposizione}
    \label{thm:mc2}
   Let  $s \in  F$ be an element of order $j\mid p^m-1$ in $F$. Then there exists a unique $r\in R$ such that $r^j=1$ and $\bar r=s.$ 
\end{proposizione}
\begin{proof}
   Since $\gcd(j,p)=1$, $x^{j}-1$ has only simple roots in $F$.  By \Cref{thm:mc1}, there exists a unique $r\in R$ such that $r^j=1$ and $\bar r=s$.  
\end{proof}
    Hence, the polynomial $x^j-1$ splits in $R$ if and only if it splits in  $F$.  \\

Since $x^{p^m-1}-1$ splits in $F$, the following Proposition is a consequence of  \Cref{lem:projss}.
\begin{proposizione}\label{prop:maxsubwc}
    Let $q=p^m-1$ and let $g\in R$ be a primitive $q$th root of unity. Then $g^i-g^j$ is a unit for all $0\le j< i \le q-1$.
\end{proposizione}
Let  $G$  be the cyclic subgroup of $\nr$ whose elements are the roots of the polynomial $x^{p^m-1}-1\in \rx$. \Cref{prop:maxsubwc} implies
that $G$ is a subtractive subset in $\nr$. \Cref{lem:projss} implies that the size of any subtractive subset of $\nr$ cannot exceed $p^m-1$. A subtractive subset of  $\nr$ of  size $p^m-1$ is said to be a \emph{maximal subtractive subset}.  Thus, $G$ is a maximal subtractive subset of $R$.
\begin{proposizione}\label{prop:goodan}
    Let $H$ be a subgroup of the cyclic group $G$. The annihilator polynomial of the subgroup \begin{equation*}
        p(x)=\prod_{h\in H}(x-h)=x^{\mo H}-1 \ ,
    \end{equation*} is constant on the cosets of $H$.
\end{proposizione}
\begin{proof}
    Let $ a\bar h$, $\bar h\in H$ be two elements in the coset $aH$. \begin{equation*}
        p(a\bar h)= \prod_{h\in H}(a\bar h-h)=\bar h^{\mo{H}}\prod_{h\in H}(a-h\bar h^{-1})=\prod_{h\in H}(a-h)=p(a) \ . 
    \end{equation*}
\end{proof}
\begin{rem}
 We can choose $p(x)=x^{\mo H}$ instead of dealing with $p(x)=x^{\mo H}-1$.
\end{rem}

The annihilators of subgroups form a class of $\big(\mo{H}-1,(p^m-1)/\mo H\big)$-good polynomials that can be employed in constructing optimal codes. \\ Since the size of the subgroup $H$ divides the size of the group $G$, $p^m \equiv 1 \mod r+1$ and the length of the code is always a  multiple of $r+1$. It is worth highlighting that the sizes of the possible subgroups  and maximum size of a subtractive subset impose constraints on the parameters of the code.     \\

\begin{rem}
    Analogously to \cite[Proposition 4.3]{tamo}, by selecting two distinct subgroups of $G$ with coprime orders, we can construct locally recoverable codes with two disjoint recovery sets.
\end{rem}

\section{A generalized version of the previous construction}\label{sec:5}
\subsection{Algebra of good polynomials over finite chain rings}Let $R$ be a finite chain ring,  let $\gamma$ be the generator of the maximal ideal and let $s$ be its  nilpotency index. Let $A$ be a well-conditioned set of size $n$ and let $A=\bigcup_{i=1}^lA_i$ be a partition of $A$. Let
\begin{equation}\label{eq:algpoly}
    \fa=\{ f\in \rx \mid  f(A_i)=c_i \forall \ i\in\{1,\dots,l\}\ ,   \text{deg }f <\mo{A}  \} \end{equation} be the set of polynomials   over $R$  of degree less then $\mo{A}$ constant on blocks of the partition. 
    \begin{definizione}
        The \emph{annihilator polynomial} of $A$ is the monic polynomial of smallest degree $h$ such that $h(a)=0$ for all $a\in A$. 
    \end{definizione}
  We  endowed $\fa$ with the  multiplication modulo $h$: $$\fa\times \fa \to \fa \ , \quad (f,g)\mapsto fg \mod h \ . $$  We can observe that: \begin{itemize}
     \item  $\fa $  is a commutative  ring; 
     \item  $\fa$ is an $R$-module;
     \item The ring product is compatible with the module product, namely, the scalar multiplication is bilinear: $$ r\cdot (fg)= (r\cdot f)g=f(r\cdot g) \ .$$
 \end{itemize}
 Therefore $\fa$ with the usual addition and the multiplication modulo $h$ is a commutative algebra over $R$. We now investigate some properties of  $\fa$.
 \begin{proposizione}\label{prop:propal} The following holds for $\fa$:
     \begin{enumerate}
         \item If $f\in\fa$ is a non-constant polynomial then $\max_i \mo{A_i}\le \deg f <\mo A$;
         \item $\fa$ is a free algebra of dimension $l$, namely, $\fa$ is a free $R$-module with basis $\{f_1,\dots,f_l\}$ with $f_i(A_j)=\delta_{i,j}$ and deg $f_i<\mo A$ (where $\delta_{i,j}$ is the Kroneker delta). Explicitly, $$f_i(x)=\sum_{a\in A_i} \prod_{b\in A\setminus\{ a\}} \frac{x-b}{a-b} \ ; $$
         \item Let $\{c_1,\dots, c_l\}$ be a well-conditioned set in $R$ and let $g$ be the polynomial of degree less than $\mo A$ satisfying $g(A_i)=c_i$ for all $1\le i \le l$, i.e. , $$g(x)=\sum_{i=1}^lc_i\sum_{a\in A_i}\prod_{b\in A\setminus\{ a\}} \frac{x-a}{b-a} \ .$$ Then the polynomials $\{1,g,\ldots,g^{l-1}\}$ form a basis for $\fa$.
     \end{enumerate}
 \end{proposizione}

 \begin{proof}
     \begin{enumerate}
         \item Let $c\coloneqq f(A_i)$. The polynomial $f(x)-c$ has at least $\mo{A_i}$ roots in the well-conditioned set $A$. Therefore  $\deg f\ge \mo{A_i}$ for all $1\le i\le l$, and hence the claim.
         \item Since $A$ is a well-conditioned set,  the polynomials $f_i(x)$ are well-defined for all $i$. By definition, the set of polynomials $\{f_1,\dots,f_l\}$ generate $\fa$. Moreover the $f_i$s are linearly independent: if $\sum_{i=0}^l\lambda_if_i(x)=0$ then \\$\sum_{i=0}^l\lambda_if_i(A_j)=\sum_{i=0}^l\lambda_i\delta_{i,j}=\lambda_j=0$. Therefore, $\{f_1,\dots,f_l\}$  form a basis for $\fa$.
         \item 
         If $
             \sum_{j=1}^l b_jg(x)^{j-1}=0 \  
        $ implies $b_j=0$ for all $1\le j \le l$ then $1,g,\dots,g^{l-1}$ are linearly independent.  Notice that the equation $ 
             \sum_{j=1}^l b_jg(x)^{j-1}=0 \ , 
        $  is equivalent to the system $$V(b_1,\dots,b_l)^\top=0 \ ,$$ where $V=(g^{j-1}(A_i))_{1\le i,j\le l}$ is a Vandermonde matrix. Since  $\{c_1,\dots,c_l\}$ is a well-conditioned set  in $R$, $ \det V=\prod_{i\ne j}(c_i-c_j)$  is a unit in $R$, hence $V$ is of full rank and $V(b_1,\dots,b_l)^\top=0$ if and only if $b_j=0$ for all $j$.
          Since $\fa$ has dimension $l$,  $\{1,g,\dots,g^{l-1}\}$ generate $\fa$. 
     \end{enumerate}
 \end{proof}

 In the following, we will focus on  the algebra of \good polynomials $\fa$ where $A$  is partitioned into blocks of size $ r+1$. 

  \begin{rem}
If  $g\in\fa$  is monic  polynomial of degree $r+1$, then $g$ always takes different values on the blocks of the partition of $A$. Otherwise, for some constant $c\in R$ the non-zero polynomial $g(x)-c$ would have $2(r+1)$ roots in a well-conditioned set. Moreover if $g$ takes values $c_1,\dots,c_l$  on $A_1,\dots,A_l$  respectively, then  $\{c_1,\dots,c_l\}$ is a subtractive subset of $R$. By contradiction, assume  $g(A_i)=c_i$ and $g(A_j)={c_i+\lambda \gamma}$ for some $\lambda\in R$. Then the polynomial $h(x)=\gamma^{s-1}\big(g(x)-c_i\big)$ is a  non-zero polynomial of degree  $r+1$ having at least $2(r+1)$ roots in a well-conditioned set, which leads to a  contradiction. From  \Cref{prop:propal}(3) follows that the algebra $\fa$ is generated by the powers of  $g$. Hence, if $g$ is the good  polynomial  introduced in \Cref{teo:tbcodes} then its powers span the algebra of \good polynomials.  \end{rem}

\subsection{A family of Locally Recoverable Codes} 
Let $A$ be a well-conditioned set of size $n$ and let $A=\bigcup_{i=1}^l A_i $ be a partition of $A$ into $l$ subssets of size $r+1$. Let 
$$\fa^r=\bigoplus_{i=0}^{r-1} \fa x^i \ .$$ Note that $\fa^r$, being the direct sum of algebras of dimension $l$, has dimension $lr$.\\
The idea behind the next construction is to associate  in a injective way the messages $a\in R^K$ to  polynomials $f_a(x)\in \fa^r$, and then evaluate $f_a $ in the points of $A$.

\begin{teorema}\label{thm:tbgeneralizzata}
     Let $A_1,\dots,A_l$ be subsets of $R$ such that $A=\bigcup_{i=1}^l A_i$ is well-conditioned. Let $r\ge 1$ and assume there exists a polynomial $g\in \fa$ of degree $r+1$ whose powers span $\fa$. Let  $K=rt$ and let  $$\Phi\colon R^K\to \fa^r , \quad a\mapsto f_a(x) \ , $$ be an injective map. If we define the code  as    \begin{equation*}
   \C=\bigg\{ (f_a(\alpha), \ \alpha\in A) \mid a\in R^K\bigg\} \ , 
    \end{equation*}
    then $\C$ is a free $(n,K,r)$-code with minimum distance  $$d\ge n-\max_{a,b\in R^K} \deg (f_a-f_b)\ge n-\max_{a\in R^K} \deg f_a \ .$$
\end{teorema}
\begin{proof}
    In order to determine the parameters of the code, we essentially repeat the proof of \Cref{teo:tbcodes}. We  explicitly determine the bound on the minimum distance. For a given message vector $a\in R^K$ the encoding polynomial reads $$f_a(x)=\sum_{i=0}^{r-1} f_i(x)x^i\ ,$$ with $f_i(x)\in \fa$. Since $\{1,g,\dots g^{l-1}\}$ is a basis for $\fa$, \begin{align*}
        \deg f_a&\le (r+1)(l-1)+(r-1)\le rl+l-r-1+r-1=\\&=\frac{nr}{r+1}+\frac{n}{r+1}-2=n-2<n \ .
    \end{align*}
Let $c_a=(f_a(\alpha), \alpha\in A)$ and $c_b=(f_b(\alpha), \alpha\in A)$ be two codewords corresponding to the distinct message vectors $a$ and $b$.
Since $\Phi$ is injective and $\deg (f_a-f_b)<n$,  $c_a$ and $c_b$ are distinct and the claim follows.
\end{proof}
Notice that the recovering procedure follows the same steps of  Construction \ref{teo:tbcodes}.
\section{Removing  the constraints on code length}\label{sec:6}

\subsection{Codes  over well-conditioned sets with arbitrary length} Constructions in \Cref{teo:tbcodes}, \ref{thm:tbgeneralizzata} and \ref{thm:multiblocks} require the assumption that  $ r+1$ divides  $n$. We provide a different construction in order to relax this condition.  Let $h_A(x)=\prod_{a\in A} (x-a)$ be the annihilator polynomial of the set $A$. For  simplicity assume $r\mid K+1$ (this constraint can be easily lift, see \Cref{rem:rdivk} ). \\


Let  $g\in \fa$  be a polynomial of degree $r+1$ whose powers span $\fa$. Without loss of generality we may assume that $g$ vanishes on $A_l$, otherwise one can consider the powers of the polynomial $g(x)-g(A_l)$.

\begin{teorema}\label{thm:almopt}
    Let $A=\bigcup_{i=1}^l A_i$ be a well-conditioned set  with  $\mo{A_i}=r+1$ for all $1\le i \le l-1$ and  $\mo{A_l}=m<r+1$. Let $g(x)$ be a polynomial of degree $r+1$ whose powers span $\fa$. 
    Let $a=(a_0,\dots,a_{r-1})\in R^K$ be the message vector with $a_i\in R^{\frac{K+1}{r}}$ for  $i\ne m-1$ and $a_{m-1}\in R^{\frac{K+1}{r}-1}$. We define the encoding polynomial as $$f_a(x)=\sum_{i=0}^{m-2}\sum_{j=0}^{\frac{K+1}{r}-1}a_{i,j}g(x)^jx^i+\sum_{j=1}^{\frac{K+1}{r}-1}a_{m-1,j}g(x)^jx^{m-1}+\sum_{i=m}^{r-1}\sum_{j=0}^{\frac{K+1}{r}-1}a_{i,j}g(x)^jh_{A_l}(x) \ .$$ Let\begin{equation*}
   \C=\bigg\{ (f_a(\alpha), \ \alpha\in A \mid a\in R^K\bigg\} \ .
    \end{equation*}
    Then $\C$ is a free $(n,K,r)$-LRC code with minimum distance \begin{equation*}
        d\ge n-K-\bigg\lceil\frac{K}{r}\bigg\rceil+1
    \end{equation*}
\end{teorema}

\begin{proof}
    The degree of the encoding polynomial $f_a(x)$ is at most $$\bigg(\frac{K+1}{r}-1\bigg)(r+1)+r-1=K+1+\frac{K+1}{r}-1+r-1\le K+\bigg\lceil\frac{K}{r}\bigg\rceil +1 \ .$$ Since the encoding is linear, the bound on the minimum distance follows.\\
    For any position $\alpha \in \bigcup_{i=1}^{l-1}A_i$, the recovery procedure  follows the same steps of \ref{thm:tbgeneralizzata}. Indeed $f_a(x)\in \bigoplus_{i=0}^{r-1}\fa x^i$, and hence, any symbol can be recovered by accessing  $r$ symbols. The only specific situation worth examining is when the symbol $\alpha$ to be recovered belongs to $A_l$. It is essential to note that the polynomial $f_a(x)$ restricted to $A_l$ has degree at most $m-2$. Therefore, in order to recover the value of $f_a(\alpha), \ \alpha\in A_m$, we  find the decoding polynomial $\delta(x)=\sum_{i=0}^{m-2} f_i(\alpha)x^i$. $\delta(x)$ is obtained  from the set of $ m-1$ values $f_a(\beta)=\delta(\beta), \ \beta\in A_m\setminus\{\alpha \}\ .$  Finally we compute $f_a(\alpha)=\delta(\alpha)$.
\end{proof}

Note that the  minimum of the code $\C$ in \Cref{thm:almopt} distance is at most one less the maximum possible value. 

\subsection{LRC codes from MDS codes with arbitrary parameters}
In the following, we will construct a code such that its symbols can be partitioned into $t$ MDS codes $C_i$ of length $n_i$ and rank $K_i$.

\begin{definizione}
    Let $C$ be a code whose coordinates are partitioned into $l$ sets $A_i$ of size $n_i$. Let $C_i$ be the code restricted to the coordinates in $A_i$. The code $C$ has $(r,\rho)$-\emph{locality} if   for all $1\le i \le l$ we have 
    \begin{itemize}
        \item  $n_i\le r+\rho-1$;
        \item $d_{C_i}\ge \rho$.  
       
    \end{itemize}
   From the Singleton bound it follows that  the rank of $C_i$ is at most $r$.
\end{definizione}

On the same line of  \cite[Theorem 2.1]{kamath}, it is possible to improve bound \ref{bd:locfcr}.
\begin{teorema}
    Let $R$ be a finite chain ring and let $C$ be a linear code of length $n$, rank $K$ and with $(r,\rho)$-locality. Then 
    \begin{equation}\label{bd:rrholoc}
        d\le n-K+1-\bigg(\bigg\lceil\frac{K}{r}\bigg\rceil-1    \bigg) (\rho-1) \ .\end{equation}

\end{teorema}

\begin{teorema}
Let $r\ge 1$ and let  $A=\bigcup_{i=1}^l A_i$ be a partition of  the well-conditioned set $A$ into $l$ subsets with $\mo{A_i}=r+\rho-1$ for all $1\le i \le l$.  Let $g(x)\in\rx$ be an $(r+\rho-1,l)$-good polynomial on the blocks of the partition of $A$.   For $r\mid K$ (this constraint can be lifted, see  \Cref{rem:rdivk}), let  $a=(a_0,\dots,a_{r-1})\in R^K$ be the message vector with $a_i\in R^{\frac{K}{r}}$ for all  $1\le i\le l$. We define 
    \begin{equation}
        f_a(x)=\sum_{i=0}^{r-1}\sum_{j=0}^{\frac{K}{r}-1} a_{i,j}g(x)^jx^i \ ;
    \end{equation}
    and \begin{equation*}
   \C=\bigg\{ (f_a(\alpha), \ \alpha\in A \mid a\in R^K\bigg\} \ .
    \end{equation*}
    Then $\C$ is a  free code with $(r,\rho)$-locality and rank $K$. Moreover  $C$ is  an optimal $(r,\rho)$-LRC code.
\end{teorema}
\begin{proof}
    Suppose there is an erased symbol $f_a(\alpha)$ for some $\alpha\in A_i$. The restriction of $f_a$ to $A_i$ is a polynomial of degree  at most $r-1$. On the other hand $\mo{A_i\setminus\{\alpha\}}=r+\rho-2$ and hence $f_a(\alpha)$ can be reconstructed from any $r$ values in  the locations of $A_i$.
     Since $\C$ is linear and $\deg f_a\le\big(\frac{K}{r}-1\big)\big(\rho+r-1)+r-1=K-r+r-1+(\frac{K}{r}-1)(\rho-1)$, then $\C$ is $(r,\rho)$-optimal.
\end{proof}

The previous construction is a particular case of a more general one  based on the Chinese Remainder Theorem for rings \cite[Section V]{mcdonald}. 
\begin{definizione}
    Let $R$ be a ring. Two ideals $I$ and $J$ are called \emph{coprime} if $I+J=R$.
\end{definizione}

\begin{teorema}
    \textbf{(Chinese Remainder Theorem)} Let $R$ be a commutative ring and let $I_1,\dots,I_n$ be pairwise coprime ideals of $R$. Let $I=I_1\cap\dots\cap I_n$. The ring morphism $$\Phi\colon R/I\to R/I_1\times\dots R/I_n, \quad \ r+I\mapsto (r+I_1,\dots,r+I_n) \ ,$$ is an isomorphism. 
\end{teorema}
\begin{cor}
    Let $h_1(x),\dots,h_n(x)\in \rx$ be  pairwise coprime polynomials. Then, for any $a_1(x),\dots,a_n(x)\in \rx $, there exists a unique polynomial $f\in \rx $ of degree at most $\sum_i \deg h_i$ such that $$f(x)\equiv a_i(x) \text{ mod } h_i(x) \ \text{ for all } 1\le i \le n \ .$$
\end{cor}

Let $R $ be a finite chain ring, let $A$ be a subtractive subset  of $\nr$, and let $A=\bigcup_{i=1}^l A_i$ be a partition of $A$. The annihilator polynomials of $A_i$,\\ $ {h_i(x)=\prod_{a\in A_i} (x-a)}$ generate  pairwise coprime ideals of $\rx$.
\begin{teorema}
    Let $R$ be a finite chain ring and let $A$ be a subtractive subset of $\nr$. Let $A=\bigcup_{i=1}^l A_i$ be a partition of $A$ such that $\mo{A_i}=n_i$ for all  $1\le i \le l$. Let $$\psi\colon R^K\to \mathcal{F}_{K_1}\times\dots\times\mathcal{F}_{K_l},   \quad a\mapsto (a_1(x),\dots,a_l(x)) \ ,$$ be an injective mapping, where   $\mathcal{F}_{K_i}$ is the space of polynomial of degree less then $K_i$. Let $h_i(x)$ be the annihilator polynomial of  $A_i$. For any message vector $a\in R^K$ we define the encoding polynomial $f_a(x)$ as the unique polynomial of degree less then $n$ such that $$f_a(x)=a_i(x)\text{ mod } h_i(x) \ . $$ Let  \begin{equation*}
   \C=\bigg\{ (f_a(\alpha), \ \alpha\in A) \mid a\in R^K\bigg\} \ .
    \end{equation*}
    Then, $\C$ is a free LRC code of rank $K$. Moreover $\C$ can be partitioned into  $l$ disjoint local codes $\C_i$, where $\C_i$ is an $(n_i,K_i)$-MDS code. 
\end{teorema}
\begin{proof}
Let $f_a(x)$ be  the encoding polynomial of the message vector $a$.  By construction, for all $1\le i\le l$, there exists a polynomial $g(x)$ such that $$ f_a(x)=g(x)h_i(x)+a_i(x) \ .$$ Thus,  for all $\alpha\in A_i$, $f_a(\alpha)=a_i(x)$. Hence, the restriction of $f_a(x)$ to $A_i$ is a polynomial of degree less then $K_i$. Since $\mo{A_i}=n_i$,   $ f_a(x)\big|_{A_i}$ is a polynomial of degree less then $K_i$ evaluated on $n_i>K_i$ points. Therefore the set of codewords $$(f_a(\alpha), \alpha\in A_i) $$  form an $(n_i,K_i)$-MDS code for all $1\le i \le l $. 
\end{proof}
We observe that the minimum distance of the code constructed in this way is at least the minimum between the distances of the local codes $\C_i$. 
\subsection{LRC codes with non-well-conditioned sets}
The most significant limitation in the previous approaches is the restriction on the  code length. The maximum code length coincides with the maximum size of a well-conditioned set. To address this problem, we now try to extend \Cref{teo:tbcodes} to non-well-conditioned sets.\\

For simplicity, let $R=GR(p^s,m)$  be a Galois ring and let  $N(R)$ denote its  group of units  having  size $p^{m(s-1)}(p^m-1)$.\\ 
Let $G$ be the maximal cyclic subgroup of $N(R)$ of order coprime with $p$ and let $H$ be a subgroup of $G$. The cosets $A_1,\dots,A_l$ of $H$ in $N(R)$ induce a partition of $N(R)=\bigcup_{i=1}^lA_i$. While $H$ is subtractive in $\nr$ (see \Cref{prop:maxsubwc}), the same does not hold true for  $N(R)$. However $N(R)$  contains a (maximal) subtractive subset. Up to reordering, we can assume $\mathcal{A}=\bigcup_{i=1}^m A_i$, $m<l$, to be a maximal subtractive subset in $\nr$.\\ 

The difference between \Cref{teo:tbcodes} and the next construction lies in the choice of the set of evaluation points: in the former a maximal subtractive subset is used, while in the latter the entire $N(R)$ is employed.

\begin{teorema}\label{thm:multiblocks}

    Let $r\ge 1$ and let  $N(R)=\bigcup_{i=1}^l A_i$ be a partition of $N(R)$ into $l$ subtractive subsets  $A_i$ of size $r+1$ for all $1\le i \le l$. Let $\mathcal A= \bigcup_{i=1}^m A_i$, $ m< l$, be a maximal subtractive subset of $N(R)$.  Let $g(x)\in\rx$ be an \good polynomial on the blocks of the partition of $N(R)$. For $t\le l $, set $n=(r+1)l$ and $K=rt$. Let $$a=(a_{i,j}, \ 0\le i\le r-1, \ 0\le j\le t-1)\in R^K \ .$$ We define 
    \begin{equation}\label{eq:encpomult}
        f_a(x)=\sum_{i=0}^{r-1}\sum_{j=0}^{t-1} a_{i,j}g(x)^jx^i \ ;
    \end{equation}
    and \begin{equation*}
   \C=\bigg\{ (f_a(\alpha), \ \alpha\in N(R)) \mid a\in R^K\bigg\} \ .
    \end{equation*}
    Then $\C$ is a free $(n,K,r)$-code where $n=\mo{N(R)}=p^{m(s-1)}(p^m-1)$  and minimum distance  $$d=n-p^{m(s-1)}\bigg(K+ \frac{K}{r}-2\bigg)=p^{m(s-1)} d_{\C'} \ ,$$  where $\C'=\C\big|_\mathcal A$ is the restriction of $\C $ to the maximal subtractive subset $\mathcal A=\bigcup_{i=0}^m A_i$. 
\end{teorema}
\begin{proof}
   The proof follows the same line of \ref{teo:tbcodes}. We explicitly compute the minimum distance of the code.\\ Let $f_a(x)$ be the encoding polynomial of the message vector $a$. The  maximum number of zeros of $f_a(x)$ establishes a bound on the minimum distance of $\C$.  Notice that $ \deg f_a(x)=K+\frac{K}{r}-2$, and hence, by \Cref{cor:nroots}, $$ d\ge n-\bigg(K+\frac{K}{r}-2\bigg)p^{(s-1)m} \ .$$ We  show that equality holds. Let  $\C'$ be the code obtained from $\C$ by puncturing the last $n-p^m+1$ positions, i.e., we left with  an $(p^m-1,K,r)$-LRC code over the subtractive subset $\mathcal A=\bigcup_{i=0}^m A_i$.  Note that $f_a(x)$ has at most $K+\frac{K}{r}-2$ roots in $\mathcal A$.  Moreover,  the LRC bound \eqref{bd:locfcr} ensures that there exists  a message vector $b\in R^K$  such that its encoding polynomial  $f_b(x)$ has exactly $K+\frac{K}{r}-2$ roots in $\mathcal A$. Let $\{\bar{x}_1,\dots,\bar{x}_{K+\frac{K}{r}-2}\}$ be the set of zeros of $f_b$ in $ \mathcal A$. Since the encoding is linear, $f'(x)=p^{s-1}f_b(x)$ is the encoding polynomial associated to the message vector $p^{s-1 }b\in R^K$. If  $\{\bar{x}_1,\dots,\bar{x}_{K+\frac{K}{r}-2}\}$ are the zeros of $f_b$ in $\mathcal A$, then $\{\bar{x}_1+M,\dots,\bar{x}_{K+\frac{K}{r}-2}+M\}$ are  the zeros of $f'$ in  $\nr$. Since $\mo M=p^{m(s-1)}$, $f'$ has $p^{m(s-1)}\big(K+\frac{K}{r}-2\big)$  roots in $N(R)$ and the claim follows. 
\end{proof}
\begin{rem}
We have that:

\begin{itemize}

    \item The central problem of all the previous  constructions is to  identify families of good polynomials. Construction \ref{thm:multiblocks} does not lead to a wider class of good polynomials.   Let $N(R)=\bigcup_{i=1}^l A_i$ be a partition of $\nr$ and let  $\mathcal A= \bigcup_{i=1}^m A_i, \ m<l$, be a maximal subtractive  subset in $\nr$. Let $h(x)=\prod_{a\in \mathcal A}(x-a)$ be the annihilator polynomial of $\mathcal A$. If $g'(x)$ is an \good polynomial for the partition of  $\nr$, then $g'(x)= g(x) \mod h$ where $g(x)$ is an $(r,m)$-good  polynomial for the partition $\mathcal A$. Therefore the class of $(r,l)$-good polynomials coincides modulo $h$ to the class of $(r,m)$-good polynomials. 

     \item We have removed the constraint on the maximum code length.  Nevertheless, the code does not meet the LRC bound \eqref{bd:locfcr} and thus it is not known whether it is optimal of not.  Let $\overline{C}$ be the projection of $C$ over the residue field of $R$. $\overline{C}$ is a repetition code. Indeed, a locally recoverable code of length $p^m-1$, dimension $K$, and minimum distance $d=n-K-\frac{K}{r}+2$ is iterated $p^{m(s-1)}$ times. Therefore, $\overline{C}$ is an LRC code with multiple disjoint recovering sets, consisting of $p^{m(s-1)}-1$ recovering sets of size 1 and $p^{m(s-1)}$ of size $r$. 

\end{itemize}
\end{rem}

A natural question arises: is there  any constraint on the maximum length of a code meeting the LRC bound, as a function of alphabet size?
\subsection{Bounds on the maximum length of an LRC over finite chain rings}
In the following, we will see that the problem of determining the maximum possible length of an LRC code
over a finite chain ring is closely related to the same problem over fields.\\

Let $R$ be a finite chain ring,  let $\gamma$ be the generator of the maximal ideal and let $s$ be its  nilpotency index.  Let $F$ be the residue field of $R$, i.e. $F=R/( \gamma)$. For any $C\subseteq R^n$ we define the code $(C:r)=\{e \in R^n \mid re\in C\}$. \begin{definizione}
    To any code $C\subseteq R^n$ we associate the tower of codes  over $R$\begin{equation*}
        C=(C:\gamma^0)\subseteq\dots\subseteq(C:\gamma^i)\subseteq\dots\subseteq(C:\gamma^{s-1});
    \end{equation*}
    and its projection over $F$\begin{equation*}
        \\ \overline{C}=\overline{(C:\gamma^0)}\subseteq\dots\subseteq\overline{(C:\gamma^i)}\subseteq\dots\subseteq\overline{(C:\gamma^{s-1})}.
    \end{equation*}
\end{definizione}
\begin{proposizione}
    If $C$ is an $R$-linear code of length $n$, rank $K$, minimum distance $d$ then  $\overline{(C:\gamma^{s-1})} $ is a linear code over $F$  of length $n$, dimension $K$ and minimum distance $d$. 
\end{proposizione}
For a proof see \cite[Theorem 4.2 and Theorem 4.5]{norton}.
\begin{proposizione}\label{prop:moddep}
If  $v_1\dots, v_u\in R^n$ are modularly dependent vectors in $R^n$ then $\overline{v_1},\dots,\overline{v_u}\in F^n$ are linearly dependent over $F$.
\end{proposizione}
\begin{proof}
    If $v_1,\dots,v_u$ are modularly dependent in $R^n$, i.e. , there exist $b_1,\dots,b_u\in R$, not all zero divisors, such that $\sum_{i=1}^u b_iv_i=0$. Hence $\gamma\mid \sum_{i=1}^u b_iv_i$ and $\sum_{i=1}^u \overline{b_i}\overline{v_i}=0$ with $\overline{b_i}$ not all  zero. Therefore $\overline{v_1},\dots,\overline{v_u}$ are linearly dependent over $F$. 
\end{proof}
\begin{proposizione}
    If $C$ is a locally recoverable code with locality $r$ over $R$ then $\overline{(C:\gamma^{s-1})}$ is a locally recoverable code with locality $\Tilde{r}\le r $ such that $\big\lceil\frac{K}{\Tilde{r}}\big\rceil=\big\lceil\frac{K}{r}\big\rceil$.  
 \end{proposizione}
\begin{proof}
\Cref{prop:moddep} implies that the locality of $C$ cannot increase. The claim follows from the fact that minimum distance of $C$ and $\overline{(C:\gamma^{s-1})}$  coincides.
\end{proof}

Consequently, determining the maximum possible length of the LRC code $C$ over $R$ reduces to the problem of determining the maximum  possible length of 
the code $\overline{(C:\gamma^{s-1})}$ over $F$.
   While for small code distances ($d=3, 4$) optimal LRC codes  with unbounded length   over any fixed alphabets of size $q\ge r+1$ are known, for $d\ge 5$ there is an upper bound on the
length of the optimal LRC as a function of
its alphabet size.   Guruswami et al. in \cite{Guruswami} proved that, for  $d= 5$ the length of an optimal LRC over an alphabet of size $q$ is at most  $\mathcal{O}(q^2)$. Moreover, if $d>5$  the length is at most  $\mathcal{O}(q^3)$.

\section{Conclusions}
In analogy to  codes over Finite Fields, the minimum distance of a locally recoverable code over a finite chain ring is bounded as a function of the length $n$, the rank $K$ and the locality $r$ of the code.  This   bound is tight, as we have constructed a family of evaluation codes that achieves this bound for any value of the locality parameter $r$ and with length bounded by the size of the residue field of the ring. The construction relies on the use of good polynomials as its fundamental components. Moreover,  this  construction can be extended in various directions: for instance codes over non-well-conditioned sets or  codes  with multiple recovering sets are presented.
 An interesting extension of this work would be to try to build longer locally recoverable codes. A second promising line of research would be    to build a wider class of good polynomials over finite chain rings.

\section*{Acknowledgement}
This publication was created with the co-financing of the European Union FSE-REACT-EU, PON Research and Innovation 2014-2020 DM1062/2021, and the French National Agency of Research  via the project ANR-21-CE39-0009-BARRACUDA. The authors acknowledge support from Ripple’s University Blockchain Research Initiative.
The first and third author are members of the INdAM Research Group GNSAGA. A preliminary version of this work was partially presented on  talks given at Young researcher Algebra Conference  2023 in L'Aquila, Italy  and Convegno annuale del gruppo  UMI Crittografia e Codici in  Perugia, Italy by the first author.\\
 
\bibliographystyle{plain}
\bibliography{bibliografia.bib}
\end{document}